\documentclass{sig-alternate}
\pdfoutput=1

\newcommand{\miniheading}[1]{\vspace{5pt}\noindent{\bf #1}}
\newcommand{\eat}[1]{}

\newtheorem{theorem}{Theorem}
\newenvironment{definition}[1]{\begin{trivlist}
\item[\hskip \labelsep {\bfseries Definition (\em{#1}):}]}{\end{trivlist}}
\newenvironment{notation}{\begin{trivlist}
\item[\hskip \labelsep {\bfseries Notation:}]}{\end{trivlist}}
\usepackage{color,soul}
\usepackage{paralist}
\usepackage{xcolor}
\usepackage{url}
\usepackage{fixltx2e}
\usepackage[sort]{cite}

\begin{document}
%
\conferenceinfo{(submitted for publication July 2015)}{}

\title{The Tangent Search Engine: Improved Similarity Metrics and Scalability for Math Formula Search}

\numberofauthors{2} 
\author{
\alignauthor
Richard Zanibbi\\
Kenny Davila\\
	\affaddr{Rochester Institute of Technology, USA}\\
	\email{\{rxzvcs,kxd7282\}@rit.edu}
\alignauthor
Andrew Kane\\
Frank Wm. Tompa\\
	\affaddr{University of Waterloo, Canada}\\
	\email{\{arkane,fwtompa\}@uwaterloo.ca}
}

\maketitle

\begin{abstract}

With the ever-increasing quantity and variety of data worldwide, the Web has become a rich repository of mathematical formulae. This necessitates the creation of robust and scalable systems for Mathematical Information Retrieval, where users search for mathematical information using individual formulae (query-by-expression) or a combination of keywords and formulae. Often, the pages that best satisfy users' information needs contain expressions that only approximately match the query formulae. For users trying to locate or re-find a specific expression, browse for similar formulae, or who are mathematical non-experts, the similarity of formulae depends more on the relative positions of symbols than on deep mathematical semantics.

We propose the Maximum Subtree Similarity (MSS) metric for query-by-expression that produces intuitive rankings of formulae based on their appearance, as represented by the types and relative positions of symbols. Because it is too expensive to apply the metric against all formulae in large collections, we first retrieve expressions using an inverted index over tuples that encode relationships between pairs of symbols, ranking hits using the Dice coefficient. The top-$k$ formulae are then re-ranked using MSS. Our approach obtains state-of-the-art performance on the NTCIR-11 Wikipedia formula retrieval benchmark and is efficient in terms of both index space and overall retrieval time. Retrieval systems for other graphical forms, including chemical diagrams, flowcharts, figures, and tables, may also benefit from adopting our approach.

\end{abstract}



\section{Introduction}

Mathematical Information Retrieval (MIR) is an important emerging area of Information Retrieval research~\cite{survey2012,Aizawa2013bc,AizKohOun:nmto14,wikipedia}. Technical documents often include a substantial amount of mathematics, but math is difficult to use directly in queries. For the most part, large-scale search engines do not support formula search other than indirectly, e.g., through matching \LaTeX~strings.  Formula queries allow documents with similar expressions or mathematical models to be discovered automatically, providing a new way to search and browse technical literature~\cite{Youssef2006a}.
For mathematical non-experts, querying based on the appearance of expressions may also be useful, for example when students try to interpret unfamiliar notation~\cite{Wangari2014}.
Many have had the experience of wishing they could search through technical documents for similar formulae rather than find words to describe them.

Figure~\ref{fig:q14} shows the top of a results page from the new \emph{Tangent} \cite{Stalnaker2015lq,Pattaniyil2014}  formula retrieval engine.\footnote{\url{http://www.cs.rit.edu/~dprl/Software.html}} The 17 hits shown are grouped by their structure (exact match, variable substitution, operator substitution), and groups are ordered by the similarity of the contained formulae to the query. 
Efficient and effective retrieval becomes more difficult when the best matches are even less similar to the query formula (e.g., the repository includes larger expressions that include pieces similar to one or more parts of the query formula) or when wildcards that can match arbitrary symbols or subexpressions are included in the query~\cite{Kamali:2010:NMR:1871437.1871635}.

\begin{figure}[!tb]

\centering

\fbox{
\scalebox{0.45}{\includegraphics[trim=20 5 10 10,clip]{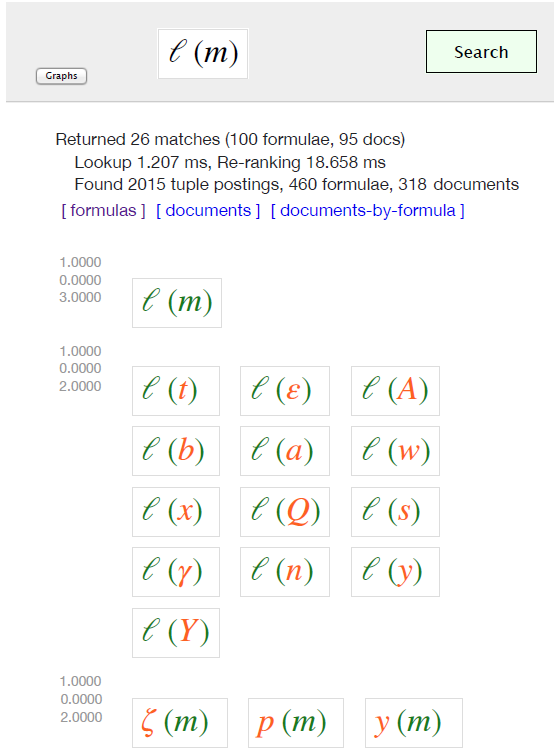}}
}
\caption{Tangent Search Results Page (truncated).} 

\label{fig:q14}
\end{figure}

For scalability, Tangent now employs a two-level cascading search system~\cite{wang2011cascade} that provides both query runtime efficiency and ranking effectiveness for formula search.  The first level is the {\em core engine}, which uses an uncompressed inverted index over tuples representing pairs of symbols in an expression. This level provides limited support for wildcard symbols and can quickly produce an ordered list of candidate results using a simple ranking algorithm.  The second level re-ranks the top candidate results using {\em Maximum Subtree Similarity} (MSS), a new metric for computing the similarity of mathematical formulae based on their appearance. The system architecture is summarized in Figure~\ref{fig:arch}.

\begin{figure*}[!tb]
\centering
\scalebox{0.39}{\includegraphics[trim=45 0 0 45,clip]{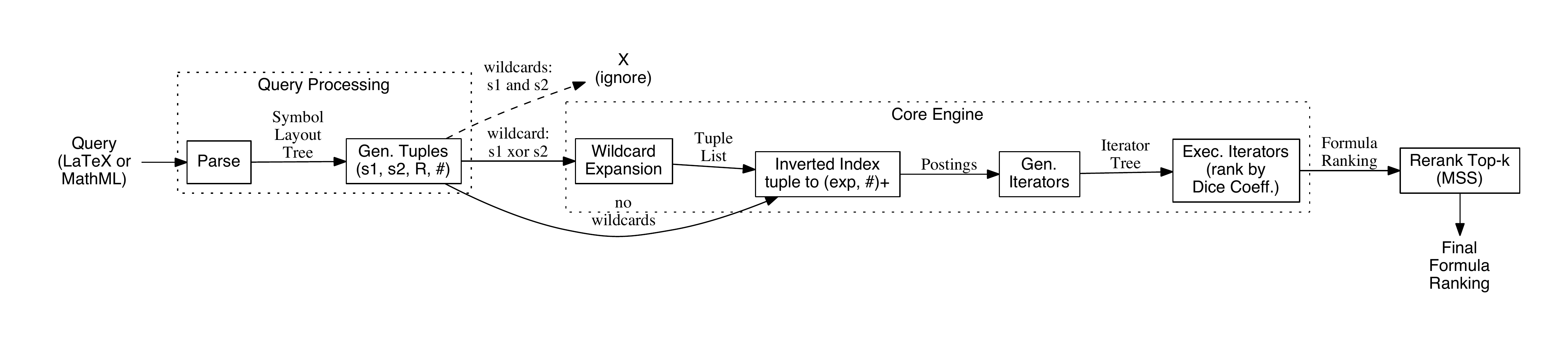}}
\vspace{-0.4in}
\caption{Formula Retrieval in Tangent (version 3)}
\label{fig:arch}
\end{figure*}

{\bf Contributions.} This paper includes three primary contributions. Our first is the incorporation of substantially smaller indices than those used previously\cite{Stalnaker2015lq,Pattaniyil2014} (Section~\ref{sec:coreengine}), which can obtain strong retrieval results in a scalable system. The second contribution is the MSS metric (Section~\ref{sec:similarity}), which produces an intuitive ordering for retrieved formula based on the visual structure of expressions, taking unifiable symbol types into account. The third is a new symbol pair retrieval model (Section \ref{sec:structure}) that incorporates the first two contributions in an efficient and effective two-stage cascaded implementation, as demonstrated experimentally (Section~\ref{sec:experiments}). 
In addition, we believe that the form of output adopted, namely grouping results by similarity and match structure, is an improvement over existing MIR interfaces.

\section{Related Work}

Interest in Mathematical Information Retrieval (MIR) has been increasing in recent years, as witnessed by the NTCIR-10~\cite{Aizawa2013bc} and NTCIR-11~\cite{AizKohOun:nmto14} Math Retrieval Tasks held in 2013 and 2014, respectively.

Math representations are naturally hierarchical, and often represented by trees that may be encoded as text strings. As a result, approaches to query-by-expression may be categorized as \emph{tree-based} or \emph{text-based}, as determined by the structures used to represent formulae.
The encoded hierarchies commonly represent either the arrangement of symbols on writing lines (as in \LaTeX~or Presentation MathML) or the underlying mathematical semantics as nested applications of operations to arguments (as in OpenMath or Content MathML). Both appearance and semantic representations have been used for retrieval.

{\bf Text-Based Approaches.}
In text-based approaches, math expression trees are linearized, and often normalized, before indexing and retrieval. Common normalizations include defining synonyms for symbols (e.g., function names), using canonical orderings for commutative operators and spatial relationships (e.g., to group {\verb a+b } with {\verb b+a } and {\verb x_i^2 } with {\verb x^2_i }), enumerating variables, and replacing symbols by their mathematical type (e.g., numbers, variables, and classes of operators)~\cite{sojka2011, survey2012}.

Although linearization masks significant amounts of structural information, it allows text and math retrieval to be carried out efficiently by a single search engine (commonly Lucene\footnote{\url{https://lucene.apache.org/}}). As a result, most text-based formula retrieval methods use TF-IDF (term frequency-inverse document frequency) retrieval after linearizing expressions~\cite{Miller2003, sojka2011}.
In an alternative approach, the largest common substring between the query formula and each indexed expression is used to retrieve \LaTeX~strings
\cite{Kumar2012a}. This captures more structural information, but also requires evaluating all expressions in the index using a quadratic algorithm.

{\bf Tree-Based Approaches.}
Tree-based formula retrieval approaches use explicit trees to represent expression appearance or semantics directly.
These approaches index complete formula trees, often along with their subexpressions to support partial matching.
Methods have been developed to compress tree indices by storing identical subtrees
uniquely~\cite{Kamali:2010:NMR:1871437.1871635} and to match expressions using tree-edit distances with early stopping for fast retrieval~\cite{kamali2013structural}. The \emph{substitution tree} data structure, first designed for unification of predicates~\cite{graf1995}, has been used to create tree-structured indices for formulae~\cite{kohlhase2006}. Descendants of an index tree node contain expressions that unify with the parameterized expression stored at that node.

A recent tree-based technique adapts TF-IDF retrieval for vectors of subexpressions and generalized subexpressions in which arguments are represented by untyped placeholders~\cite{Lin:2014:MRS:2600428.2609611}. In this method a Symbol Layout Tree is modified to capture some semantic properties, normalizing the order of arguments for commutative operators and representing operator precedences explicitly.

{\bf `Spectral' Tree-Based Approaches.} An emerging sub-class of the tree-based approach uses paths or small subtrees rather than complete subtrees for retrieval. One system converts sub-expressions in operator trees to words representing individual arguments and operator-argument triples~\cite{Nguyen20125820}.
A lattice over the sets of generated words is used to define similarity, and a breadth-first search constructs a neighbor graph traversed during retrieval.  Another system employs an inverted index over paths in operator trees from the root to each operator and operand, using exact matching of paths for retrieval~\cite{Hiroya2013yq}.  The large number of possible unique paths combined with exact matching make this technique brittle.

Rather than indexing paths from the root of the tree, the Tangent math retrieval system stores \emph{relative} positions of symbol pairs in Symbol Layout Trees to create a ``bag of symbol pairs'' representation~\cite{Pattaniyil2014,Stalnaker2015lq}. This symbol pair representation supports partial matches in a flexible way, while preserving enough structural information to return exact matches for queries. Set agreement metrics are applied to the bags of symbol pairs to compute formula similarities. For example, the harmonic mean for the percentage of matched pairs in the query and a candidate (i.e., \emph{Dice's coefficient} for set similarity\footnote{Given a query tree $T_q$ and a candidate tree $T_c$, let $F_q$ and $F_c$, respectively, denote a set of their features (such as a set of node and edge labels) and let $F_{q,c} = F_q \cap F_c$ denote the set of features they have in common. Dice's coefficient of similarity ($\frac{2|F_{q,c}|}{|F_q|+|F_c|}$) can then serve as the {\em score} for $T_c$.}) prefers large matches of the query with few additional symbols in the candidate. Tangent (starting with Version 2) also accommodates matrices, isolated symbols, and wildcard symbols and augments formula search with text search.
Formula retrieval based on bags of symbol pairs combined with keyword retrieval using Lucene allowed Tangent to produce the highest Precision@5 result for the NTCIR-11 Math-2 main retrieval task with combined text and formula queries (92\%)~\cite{AizKohOun:nmto14}.

In this paper, we address needed improvements for Tangent, described in the next section.

\section{Problem Statement}
\label{sec:problem}

The math retrieval task we address is to search a corpus to produce a ranked list of formulae (and the pages on which those formulae are located) that match a query formula expressed in \LaTeX~or Presentation MathML, with or without the inclusion of wildcard symbols. Formulae ranked highly should match the query formula exactly or, failing that, closely resemble it. The system should be scalable in terms of index size, indexing speed, and querying speed.

{\bf Scalability and Retrieval Effectiveness.} As originally implemented, Tangent is not scalable: indexing time is less than 200 formulae per second, producing indices of over 1GB for the NTCIR-11 Wikipedia corpus 
and 30 GB for the NTCIR-11 arXiv corpus.
Retrieval time is also slow, averaging 5 seconds per query for the Wikipedia task (under 400 thousand distinct formulae) and averaging 3 \emph{minutes} per query for the NTCIR main task (3 million distinct formulae). Furthermore, while retrieval effectiveness is very good, there is substantial room for improvement.

\section{Formula Structure Model}
\label{sec:structure}

\subsection{Symbol Layout Tree (SLT)} \label{sec:slt}

{\bf Symbols and Containers.} Tangent uses a Symbol Layout Tree (SLT) to represent the appearance of a mathematical formula. Tree nodes represent individual symbols and visually explicit aggregates, such as fractions, matrices, function arguments, and parenthesized expressions. In Tangent Version 3, all symbols except those representing operators or separators (e.g., commas) are prefixed with their type, represented by a single character followed by an exclamation point. More specifically, SLT nodes represent:

\begin{compactitem}
    \item{typed mathematical symbols: numbers (N!$n$); variable names (V!$v$); text fragments, such as {\it lim}, {\it otherwise}, and {\it such that} (T!$t$)}

    \item{fractions (F!)}
    \item{container objects: radicals (R!); matrices, tabular structures, and parenthesized expressions (M!$f$$r$x$c$)}
    \item{explicitly specified whitespace (W!)}
    \item{wildcard symbols (?$w$)}
        \item{mathematical operators}
\end{compactitem}

Because of their visual similarity, all tabular structures, including matrices, binomial coefficients, and piecewise defined functions are encoded using the matrix indicator M!. If a matrix-like structure is surrounded by fence characters, then those symbols are indicated after the exclamation mark. Finally, the indicator includes a pair of numbers separated by an $x$, indicating the number of rows and the number of columns in the structure. For example, M!2x3 represents a 2x3 table with no surrounding delimiters and M!()1x5 represents a 1x5 table surrounded by parentheses. Importantly, {\em all} parenthesized subexpressions are treated as if they were 1x1 matrices surrounded by parentheses, and, in particular, the arguments for any $n$-ary function are represented as if they were a 1x$n$ matrix surrounded by parentheses.

As well as associating a label (e.g., V!x) with every SLT node, every node has an associated type ({\em number}, {\em variable}, {\em operator}, etc.). A node's type is reflected in its label, usually represented by the part of the label up to an exclamation point (e.g., V!), but node labels preceded by a question mark (?) have type {\em wildcard}; a matrix node's type includes the matrix dimensions, but not its fence characters (e.g., M!2x3); and other node labels without exclamation marks have type {\em operator}.

{\bf Spatial Relationships.} Labeled edges in the SLT capture the spatial relationships
between objects represented by the nodes:
\begin{compactenum}
	\item{{\bf next} ($\rightarrow$) references the adjacent object that appears to the right on the same line}
	\item{{\bf within} ( \fbox{$\cdot$} ) references the radicand of a root or to the first element appearing in row-major order in a structure represented by M!}
	\item{{\bf element} ( $\multimap$ ) references the next element appearing in row-major order in a structure represented by M!}

	   \item{{\bf above} ( $\uparrow$ ) references the leftmost object on a higher line (e.g., superscript, over symbol, fraction numerator, or index for a radical)}
	   \item{{\bf below} ( $\downarrow$ ) references the leftmost object on a lower line (e.g., subscript, under symbol, fraction denominator)}
	   \item{{\bf pre-above} ( $\Uparrow$ ) references the leftmost object of a prescripted superscript}
	   \item{{\bf pre-below} ( $\Downarrow$ ) references the leftmost object of a prescripted subscript}

\end{compactenum}
An SLT is rooted at the leftmost object on the main baseline (writing line) of the formula it represents. Figure~\ref{fig:SLT} shows an example of an SLT, where for simplicity, unlabeled edges represent the {\em next} relationship and types other than {\em wildcard} are not displayed.

\begin{figure}[!tb]

 {
${\displaystyle\qquad\pi_{i}=2^{{\color{red} \fbox{\scriptsize ?x0}}}{\tbinom{N}{i}}}$}

~~~~~\scalebox{0.325}{\includegraphics{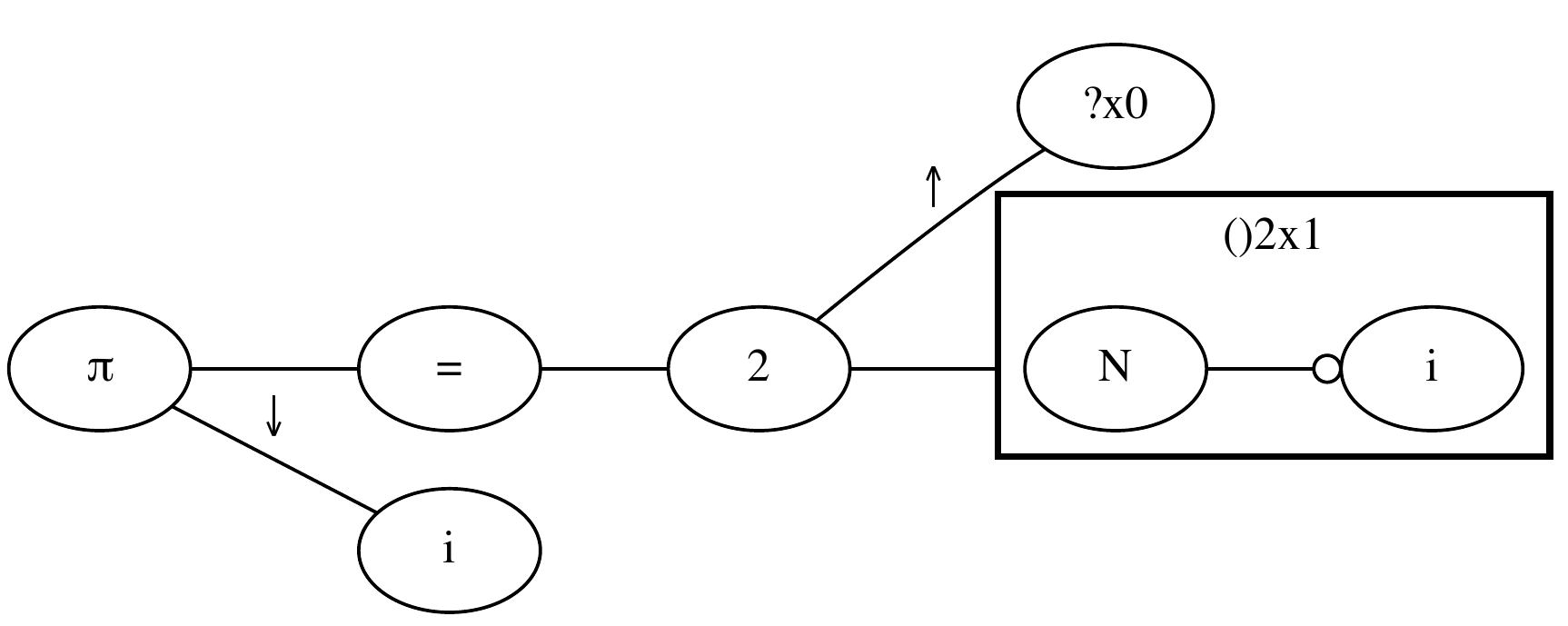}}\\
 ~\\
\small (a) Query Formula and Symbol Layout Tree (SLT)
~\\

\scriptsize
\begin{center}
\begin{tabular}{ c c l @{} c c }
 \sc Sym-1 & \sc Sym-2 & \sc Path &  \sc Count \\
\hline
{\bf V!}$\pi$ & {\bf V!}i & $\downarrow$  & 1\\
{\bf V!}$\pi$ & = & $\rightarrow$   & 1 \\
= & {\bf N!}2 & $\rightarrow$   & 1 \\
{\bf N!}2 & {\bf ?}x0 & $\uparrow$   & 1\\
{\bf N!}2 & {\bf M!}()2x1 & $\rightarrow$   & 1\\
{\bf M!}()2x1 & {\bf V!}N &  \fbox{$\cdot$}   & 1\\
{\bf V!}N &  {\bf V!}i & $\multimap$   & 1  \\

{\bf V!}N & !0 & $\rightarrow$  & 1 \\
{\bf V!}i & !0 & $\rightarrow$  & 2 \\
{\bf M!}()2x1 & !0 & $\rightarrow$  & 1\\
{\bf ?}x0 & !0 & $\rightarrow$  & 1 \\

\hline

{\bf V!}$\pi$ & {\bf N!}2 & $\rightarrow\rightarrow$  & 1 \\
= & {\bf ?}x0 & $\rightarrow\uparrow$  & 1 \\
 \vdots &\vdots & \vdots  & \vdots\\

\hline
{\bf V!}$\pi$ & {\bf V!}i & $\rightarrow\rightarrow\rightarrow$\fbox{$\cdot$}$\multimap$  & 1 \\

\hline
\end{tabular}
\end{center}
(b) Tuples for SLT Symbol Relationships with Counts

${\displaystyle\pi_{i}=2^{{-N}}{\tbinom{N}{i}}}$
\scalebox{0.305}{\includegraphics{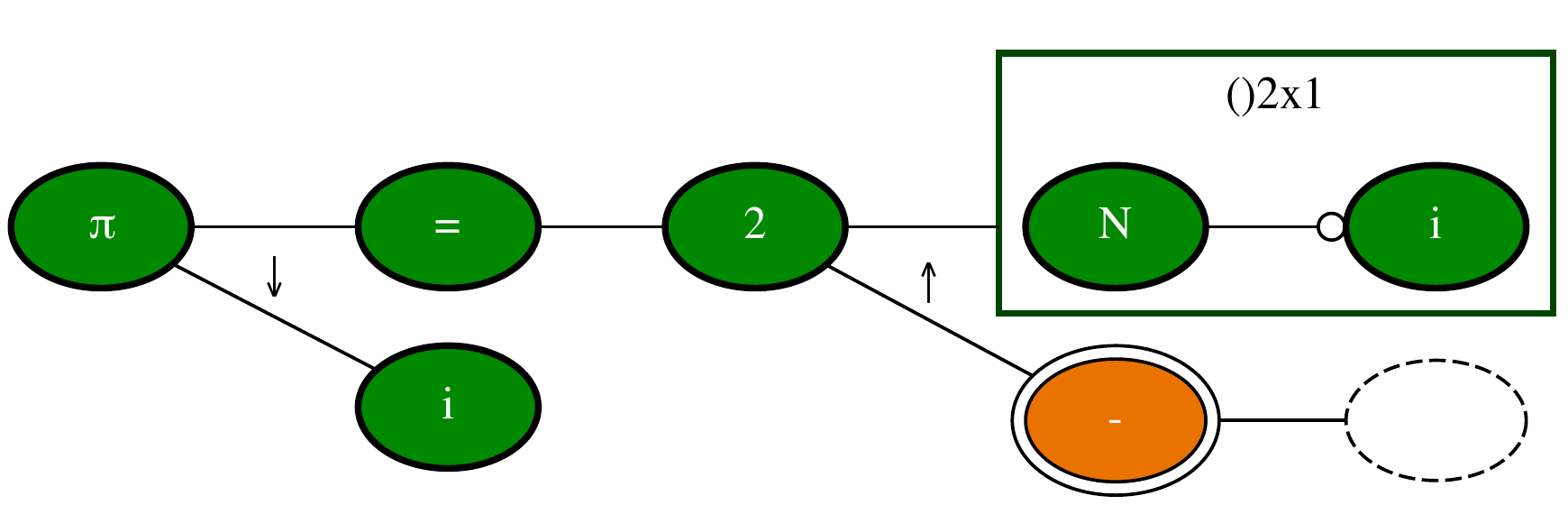}}\\

${ E_{{m,n}}=2^{{n-m}}{n\choose m}}$
\scalebox{0.305}{\includegraphics{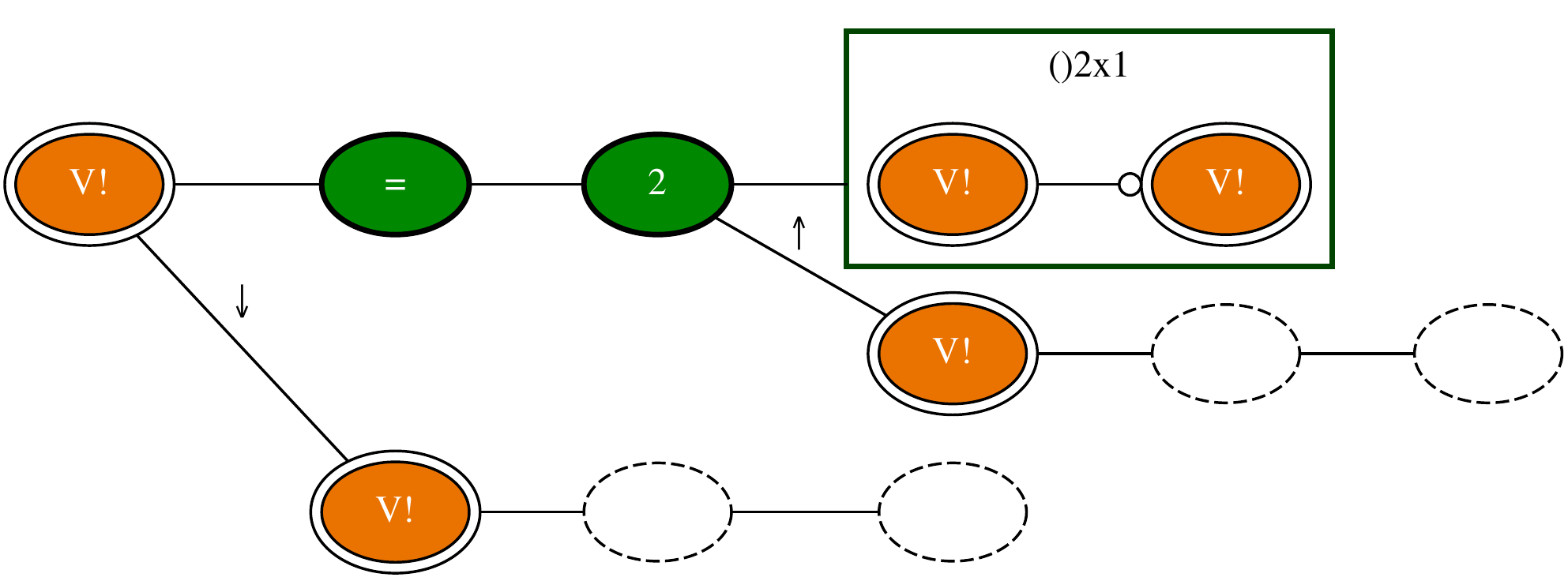}}\\
~\\
(c) Example Search Hits. Green: exact matches, Orange: unified matches, Dashed Nodes: unmatched symbols

\caption{Query Formula with Corresponding Symbol Layout Tree (SLT), Symbol Pair Tuples,
and Sample Search Results.}
\label{fig:SLT}
\end{figure}

{\bf Creating SLTs.} SLTs can be created straightforwardly from Presentational MathML by a recursive descent parser. For other input formats, we assume that converters such as LaTeXML\footnote{\url{http://dlmf.nist.gov/LaTeXML/}} exist to produce Presentational MathML.

In most circumstances, whitespace is not represented in an SLT. As a result, although unicode whitespace and related characters, such as ``invisible times'' ({\tt \small U+2062}), occasionally appear as operators in Presentational MathML expressions, they are all ignored for the purpose of matching expressions in Tangent.

\subsection{SLT Tuple Representation} \label{sec:tuples}

As described above, a node in a Symbol Layout Tree can have up to seven labeled outgoing edges (with no edge label repeating for any node).
For a given Symbol Layout Tree, Tangent produces a set of tuples that each encodes the relationship between a pair of symbols occurring on some path from the root to a leaf.
Given two nodes on such a path, we define the relative path between the nodes by the sequence of edge labels traversed from the ancestor node to the descendant.

As an optimization that saves both space and time, and following the practice of searching via n-grams~\cite{Willett79},
the new version of Tangent does not store all tuples defined, but only those for which the distance between symbols
(measured by the number of edges separating them) is less than or equal to a specified window size $w$.

In addition to normal tuples for an SLT, end-of-line (EOL) information can optionally be captured by introducing
special tuples of the form $(last~symbol, !0, \rightarrow)$.  Such end-of-line tuples are likely to improve retrieval, particularly for
individual symbols and small expressions.

Tuple information for the example expression in Figure~\ref{fig:SLT} along with tuple counts is shown
in Figure~\ref{fig:SLT}b. 
The maximum
path length between two symbols is five, yielding 23 distinct tuples (19 symbol pair tuples plus four EOL tuples with one repetition for `i'). However, if the window size $w$ is set to 2, then only 16 distinct symbol pairs are stored. 

\section{Architectural Overview}

In order to improve both query runtime efficiency and ranking effectiveness, our search system uses a two level cascading
approach~\cite{wang2011cascade}.  After parsing a query formula, the first level searches the corpus and returns candidate results. The second level then re-ranks those candidates, and finally the results are displayed on an HTML page using grouping and color coding of match structures for improved clarity.

The first level of our search system, referred to as the {\it core engine}, uses an inverted index over tuples defined
from symbol pairs of the Symbol Layout Trees of expressions.  The core engine supports limited query functionality
in order to produce a small list of candidate expression results quickly using the simple ranking metric defined
by Dice's coefficient over tuple matches.  These candidate expressions are returned together with the lists of documents that
contain each expression.

The second level of our system is a re-ranker that implements the full query functionality to identify tuple matches in expressions.
The re-ranker scores the candidate expression results using the more accurate Maximum Subtree Similarity ranking metric, as defined
in Section~\ref{sec:similarity}.  The re-ranker then combines several expression scores in the candidate expression results to produce a final document ranking.

The final query results can be ordered by either formula rank or document rank.  When the results
are presented in formula rank order, they are grouped by their maximum query subtree overlaps as shown in
Figure~\ref{fig:q14}; when presented in document rank order they are grouped by document.
In either case, expressions are displayed color coded to highlight their maximum subtree overlaps, as shown in Figure~\ref{fig:SLT}c.

\section{Core Engine} \label{sec:coreengine}

The core engine for our system quickly finds a small number of highly relevant candidate results for a math
search query, which are later re-ranked.
The engine returns these top-$k$ formulae determined
using a simple ranking algorithm, along with the list of documents containing each formula and the first position
of that formula in the document.

Since runtime performance is a high priority, the core engine uses a customized inverted index data structure
implemented in C++.  In addition, the engine evaluates only a subset of the query language functionality to allow
the use of a fast and simple ranking algorithm that can still find a good set of candidate results.

The input to the indexer is a set of document names and the extracted mathematical formulae found in each document:
$\{\mathit{document},\mathit{formula}^+\}^*$, and the input to the search component is a single query formula. Each formula is converted to a set of tuples
(see Section~\ref{sec:tuples}) that serve as words do in a normal search engine.

\miniheading{Index Structures:}
At index time, an inverted index is built over the given
document-formula-tuple relationships. The index includes postings lists \emph{PL1} that map each tuple to all formulae containing that tuple.
A query containing only non-wildcard tuples can thus be implemented by combining
the corresponding tuples' postings lists using an OR operator.  We store these postings lists as ordered lists
of formula identifiers (integers), so that the lists can be easily combined using a merge algorithm.
The engine uses a dictionary \emph{D1}
to always assign the same formula to the same identifier, thus saving both space and time in the engine.

In order to return document information for query results, the engine stores
postings lists \emph{PL2} mapping each formula identifier to the identifiers of the documents
containing those formulae, along with their
first position in the document.  To improve compression, a dictionary \emph{D2} is used for document names and another
dictionary \emph{D3} is used for tuples.

The core engine supports limited wildcard functionality.  Query
tuples containing a single wildcard as either the ancestor or descendant symbol are implemented as iterator
expansions.  The engine stores postings lists \emph{PL3} that map each single wildcard tuple to the set of
tuple identifiers that match.  Assigning tuple identifiers using a dictionary \emph{D4} again gives some compression benefits.  Implementing even this restricted wildcard functionality can be expensive, since the iterator expansion could be quite large.\footnote{The engine does not try to enforce wildcard variable agreement between tuples (wildcard joins), and it ignores multi-wildcard tuples.
An initial implementation handling multi-wildcard tuples and wildcard joins was found to be approximately a hundred times slower than the current engine for a small dataset.}

In summary, the core engine uses two main data structures: \emph{dictionaries} convert objects (such as strings or tuples) into a
compact 0-based range of internal identifiers (integers) and \emph{postings lists} are lists of integer tuples ordered by
the first integer in the tuple.
\begin{compactitem}
\item dictionaries
\begin{compactitem}
\item[\emph{D1}:] formula $\rightarrow$ formID
\item[\emph{D2}:] document $\rightarrow$ docID
\item[\emph{D3}:] tuple $\rightarrow$ tupleID
\item[\emph{D4}:] wildcardtuple $\rightarrow$ wildcardtupleID
\end{compactitem}
\item postings lists
\begin{compactitem}
\item[\emph{PL1}:] tupleID $\rightarrow$ $($formID, count$)^+$
\item[\emph{PL2}:] formID $\rightarrow$ $($docID, position$)^+$
\item[\emph{PL3}:] wildcardtupleID $\rightarrow$ $($tupleID$)^+$
\end{compactitem}
\end{compactitem}
These data structures can be combined to produce compression, ease of storage, and fast
access speeds.  

\miniheading{Searching:}
Query processing follows the architecture shown in Figure~\ref{fig:arch}.  First, the query is parsed into an SLT and tuples
are extracted.  Then wildcard tuples are expanded, the associated postings lists for each tuple are found,
iterators over these lists are created, and an iterator tree that implements the query is formed.  Next, the
iterator tree is advanced along formula identifiers in order, the scores are calculated, and the top-$k$
formulae are stored in a heap.  During this process, non-wildcard iterators are advanced first so that wildcard iterators
only match unallocated tuples.
As optimizations, iterators may skip over some formulae based on thresholds and
max-score calculations ({\it see below}).  After the iterators are finished, matching formulae and scores are returned
along with the associated document names.

The engine uses Dice's coefficient over tuples as a simple ranking algorithm, counting the number of tuples that overlap
between the query and a candidate formula using the query iterators.  The engine also stores the tuple
count for each formula in an array \emph{A1} and uses these values in the ranking calculation:
\begin{compactitem}
\begin{compactitem}
\item[\emph{A1}:] formID $\rightarrow$ tuplecount
\end{compactitem}
\end{compactitem}
Since wildcards can often match multiple tuples in a query and overlap with other wildcards, there could be
multiple ways to count the tuples that overlap.  The
engine implements a greedy counting approach by simply assigning the matches for tuples when each of the
iterators is advanced.

\miniheading{Parameters:}
The engine has three configuration parameters: the window size $w$ for formula-to-tuple conversion, the optional use of
end-of-line tuples, and the number of formulae $k$ to return for each query.  The runtime efficiency and
ranking effectiveness of configurations using various settings of the first two parameters with $k=100$ are examined
in Section~\ref{sec:experiments}.

\miniheading{Optimizations:}
By using dictionaries and postings lists, the engine's data structures are small enough to be run in memory for the datasets
being examined, so we do not examine additional techniques to compress these data structures here.  Nevertheless, query processing might still be slow, even though
the data structures are in memory, the ranking algorithm is fast, and the use of a dictionary avoids repeated processing of duplicate
formulae.  As a result, various techniques are employed to improve query execution time:\\

\begin{compactitem}
\item[\emph{O1}:] Avoid processing all postings by allowing skipping in query iterators.  This functionality is implemented
using doubling (galloping) search~\cite{bentley1976almost}.
\item[\emph{O2}:] Skip formulae based on size thresholds.  We use the current top-$k$ candidate list to define a
minimum score that defines minimum and maximum tuple size thresholds from the definition of Dice's coefficient.
We also improve on the effectiveness of these thresholds by reordering formula identifiers: sort the formulae by size,
split into quartiles $\{q_1,q_2,q_3,q_4\}$, and then reorder $\{q_2,\mathit{reverse}(q_1),q_3,q_4\}$.
\item[\emph{O3}:] Avoid formulae that match only wildcard tuples when the score threshold allows.  This is similar
to portions of the max-score~\cite{turtle1995query} optimization, only at a coarser granularity.
\item[\emph{O4}:] Avoid processing all wildcard tuple expansions.  If a tuple is matched to a wildcard for
the next formula, do not process the remaining iterators for this wildcard.
\item[\emph{O5}:] Process iterators for large postings lists first.  Evaluate the binary operator tree left-first and
order tree operators descending by size when possible.
\end{compactitem}
\vspace{0.1in}

Various improvements to the engine have been left for future work, including compression of the postings lists and
implementing more of the query functionality in the engine.  Additional improvements in query runtime are also
possible by using an implementation of weak-AND~\cite{broder2003efficient} or a more fine-grained implementation
of max-score~\cite{turtle1995query}.

\section{Reranking by Maximum Subtree Similarity} \label{sec:similarity}

Effective information retrieval depends on ranking documents based on their similarity to a user's query. For example, when using tree-based formula retrieval, one could extract a set of features from a query tree $T_q$ and each candidate indexed tree $T_{c_i}$, apply Dice's coefficient of similarity as the score for $T_{c_i}$, and rank candidates by their scores. In this section we describe an alternative to Dice's metric that is particularly effective in ranking mathematical formulae.

\begin{notation}
The label on node $n$ in SLT $T$ is denoted $\lambda(n)$. The number of nodes in SLT $T$ is denoted $|T|$. For simplicity, we write $n \in T$ if $n$ is a node in $T$ and $(n_1,n_2) \in T$ if $(n_1,n_2)$ is an edge in $T$.
\end{notation}

Approximate matches of formulae might involve isolating corresponding parts of SLTs representing a query and a candidate match. Therefore we need a basis for describing such a correspondence.

\begin{definition}{aligned SLTs}
SLTs $T_1$ and $T_2$ are {\em aligned} if there is an isomorphism $f$ mapping nodes from $T_1$ onto nodes from $T_2$ such that for every edge $(n_a,n_b) \in T_1$, there is a corresponding edge $(f(n_a),f(n_b)) \in T_2$ that has the same label. (Note that {\em node} labels in aligned trees need not match.) For $N$ a subset of nodes in $T_1$, we define $f(N) = \{f(n) \;|\; n \in N\}$.
\end{definition}

Approximate matches might also involve simple replacements of symbols in one SLT by alternative symbols (e.g., $x$ for $y$ or 3 for 2). Naturally, a wildcard symbol can be replaced by any symbol.

\begin{definition}{unified nodes}
Node $n_1$ in SLT $T_1$ can be {\em unified} with node $n_2$ in SLT $T_2$, denoted $n_1 \dashrightarrow n_2$, if one of the following conditions holds:
\begin{compactitem}
\item{Both $n_1$ and $n_2$ have type {\em variable name} (V!),}
\item{Both $n_1$ and $n_2$ have type {\em number} (N!),}
\item{$n_1$ has type {\em wildcard} (?), or}
\item{$n_1$ has a type other than variable name, number, or wildcard and $\lambda(n_1) = \lambda(n_2)$.}
\end{compactitem}
\end{definition}

However the SLT for an arbitrary query formula will not necessarily align with the SLT for an arbitrary candidate match formula. Therefore, we need to consider parts of the SLTs that can be aligned. When considering many candidate match trees, we are most interested in those parts of the query and candidate trees that are similar to the tree representing the whole query.

\begin{definition}{maximally similar subtree}
Given SLTs $T_q$ and $T_c$ and aligned SLTs $T_1$ and $T_2$ with isomorphism $f$ from $T_1$ to $T_2$, where $T_1$ is a pruned subtree\footnote{Given a tree $T$, a {\em pruned subtree} is any connected subset of nodes from $T$ together with the edges connecting those nodes. Thus a pruned subtree is itself a tree, but it need not extend to the leaves of $T$. Henceforth, we will use ``subtree'' to mean ``pruned subtree.''} of $T_q$ and $T_2$ is a pruned subtree of $T_c$, let $m = |\{n \in T_1 \;|\; n \dashrightarrow f(n)\}|$. Let $\frac{2m}{|T_1|+|T_q|}$ be a measure of similarity of $T_1$ to $T_q$ with respect to $T_2$ (Dice's measure). $T_1$ is then {\em maximally similar} to $T_q$ if the root of $T_1$ can be unified with the root of $T_2$ and there is no other pair of aligned SLTs $T'_1$ and $T'_2$ with corresponding measure $m'$, where $T'_1$ is a subtree of $T_q$, $T'_2$ is a subtree of $T_c$, the root of $T'_1$ is the same as the root of $T_1$, the root of $T'_2$ is the same as the root of $T_2$, $|T'_1| > |T_1|$, and $\frac{2m'}{|T'_1|+|T_q|} > \frac{2m}{|T_1|+|T_q|}$.
\end{definition}

\begin{theorem}
Given SLTs $T_q$ and $T_c$ and aligned SLTs $T_1$ and $T_2$ with isomorphism $f$ from $T_1$ to $T_2$, where $T_1$ is a subtree of $T_q$ and $T_2$ is a subtree of $T_c$, determining that $T_1$ is maximally similar to $T_q$ can be performed in time $O(|T_q|)$.
\end{theorem}
\begin{proof}
Let $r$ be the root of $T_1$. The maximally similar subtree to $T_q$ and rooted at $r$ can be determined in time $O|T_q|$: \\
{\em base case:} If $r$ is a leaf of $T_q$, then $|T_1| = |T_2| = 1$ and $T_1$ is maximally similar to $T_q$ iff $r \dashrightarrow f(r)$. This can be determined in $O(1)$ time.\\
{\em recursion:} Let $\mathcal{T} = \{t_i \;|\; t_i$ is a maximally similar subtree of $T_q$, the root of $t_i$ is a child of $r$, the root of the aligned SLT for $t_i$ is a child of $f(r)$, and the label on the edge from $r$ to the root of $t_i$ is the same as the label on the edge from $f(r)$ to the root of the aligned SLT for $t_i\}$. (This construction is unambiguous because the edge labels on all edges starting at a node are unique.) The isomorphism $f$ can then be extended to include all the nodes in all subtrees in $\mathcal{T}$, and the subtree of $T_q$ consisting of $r$ and all the subtrees in $\mathcal{T}$ will be aligned with the subtree of $T_c$ having nodes $\{f(n) \;|\; n = r \vee n \in t_i \wedge t_i \in \mathcal{T}\}$. Let $m_i = |\{n \in t_i \;|\; n \dashrightarrow f(n)\}|$. The maximally similar subtree to $T_q$ and rooted at $r$ then includes $\{n \;|\; n = r \vee \exists t_i \in \mathcal{T} (\frac{2m_i}{|t_i|+1+|T_q|} > \frac{2}{1+|T_q|} \wedge n \in t_i)\}$ iff $r \dashrightarrow f(r)$ (i.e., we can evaluate the similarity for each subtree independently to determine whether or not it is part of the maximally similar subtree). Because the outdegree of $r$ is bounded by a constant, each step of the recursion can be performed in $O(1)$ time.\\
$T_1$ is then maximally similar to $T_q$ iff it is the tree thus constructed and $r \dashrightarrow f(r)$, and the construction can be performed in time $O(|T_q|)$.
\end{proof}

Next, when matching with substituted symbols, it is important that the substitutions are consistent when determining that two formulae match approximately.

\begin{definition}{alignment partition}
Given $T_1$ and $T_2$, two aligned SLTs with isomorphism $f$ from $T_1$ to $T_2$, an {\em alignment partition} is a subset of nodes $N$ in $T_1$ such that
$(x \in N \wedge y \in N) \Rightarrow (\lambda(x) = \lambda(y) \wedge \lambda(f(x)) = \lambda(f(y)) \wedge x \dashrightarrow f(x))$.
For node $n \in T_1$, we define $P(n)$ to be the alignment partition containing $n$ if it exists and $\emptyset$ otherwise. (Note that $n \in P(n) \Leftrightarrow n \dashrightarrow f(n)$.) For alignment partition $A$, $\lambda(A)$ denotes the label that is common to all nodes in $A$ and $\lambda(f(A))$ denotes the label that is common to all nodes in $f(A)$.
\end{definition}

\begin{definition}{matched set of nodes}
Given aligned SLTs $T_1$ and $T_2$ with isomorphism $f$ from $T_1$ to $T_2$ and the set of all corresponding alignment partitions, we define a {\em matched set of nodes} $M$ as
\begin{eqnarray*}
\lefteqn{M = \{ n \in T_1 \;|\; n \in P(n) \wedge \forall n' \in M }\\
& & ([\lambda(n')=\lambda(n) \vee \lambda(f(n'))=\lambda(f(n))] \Rightarrow n' \in P(n))\}
\end{eqnarray*}
In preparation to preferring matches of large {\em connected} parts of SLTs, let $E(M) = \{(n_1,n_2) \;|\; n_1 \in M \wedge n_2 \in M \wedge (n_1,n_2) \in T_1\}$.
\end{definition}

We need to accommodate situations in which symbol $x$ in the query formula is replaced by symbol $y$ in some parts of a candidate matching formula and by other symbols elsewhere, and where superfluous instances of $x$ or $y$ might appear in the candidate match.
We suggest the following properties for a scoring function, as illustrated in Figure~\ref{fig:MSSrank}: alignments with more matched symbols in close proximity to each other score higher than those with fewer matched symbols or more disconnected matches; if two candidates score equally with respect to matched symbols and their proximity, the one with fewer superfluous symbols scores higher; and everything else being equal, alignments with identical node labels score higher than alignments with distinct node labels that can be unified. Tangent uses such a scoring function:

\begin{figure}
\small
\scalebox{0.35}{\includegraphics[trim=0 15 0 0,clip]{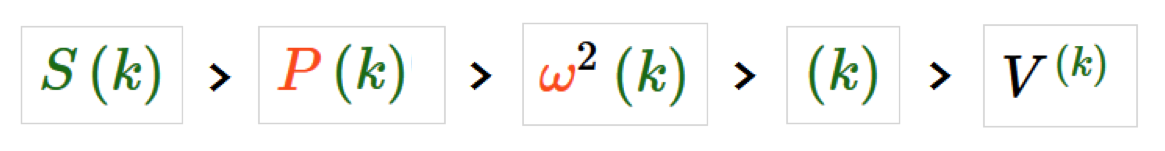}}
\mbox{~~}(1, 0, 3) \mbox{~~~~~} (1, 0, 2) \mbox{~~~~~~} (1, -1, 2) \mbox{~~~~} (0.6, 0, 2)
\mbox{~} (0.6, -1, 2)

\caption{Maximum Subtree Similarity Scoring for Query $S(k)$.}
\label{fig:MSSrank}
\end{figure}

\begin{definition}{SLT score}
Given a query SLT $T_q$, an SLT $T_c$ for a candidate match, and two aligned SLTs $T_1$ and $T_2$ where $T_1$ is a subtree of $T_q$ and $T_2$ is a subtree of $T_c$, let $M$ be a matched set of nodes for $T_1$ and $T_2$. The {\em score} of $T_c$ with respect to $T_q$, $T_1$, $T_2$, and $M$ is denoted $s(T_q,T_c;T_1,T_2,M)$ and defined as the triple composed of the following parts:
\begin{compactenum}
\item{the harmonic mean of the fraction of nodes from $T_q$ preserved by $M$ and the fraction of edges preserved by $E(M)$, i.e., $h_s = \frac{2}{\frac{|T_q|}{|M|} + \frac{|T_q|-1}{max(|E(M)|,0.5)}}$ if $|M| > 0$, otherwise 0.}
\item{the negation of the number of unmatched nodes in $T_c$, i.e., $|M| - |T_c|$.}
\item{the number of nodes that match exactly, i.e., $|\{n \in M \;|\; \lambda(n) = \lambda(f(n))\}|$.}
\end{compactenum}
The scores (triples) assigned to any two candidate matches can be computed in $O(1)$ time if $T_1$, $T_2$, and $M$ are given, and they can be compared lexicographically to determine which candidate ranks higher.
\end{definition}

For aligned SLTs $T_1$ and $T_2$ with isomorphism $f$ from $T_1$ to $T_2$ and the set of all corresponding alignment partitions, we would like to choose a matched set of nodes $M$ that produces a high score, but evaluating all matched sets induced by an alignment is too expensive. Therefore we use a greedy algorithm to select which partitions to include in the matched set of nodes, based on the properties we use for scoring:
\begin{compactenum}
\item{Let $A_0$ be the alignment partition that contains the most nodes; or if more than one partition has the most nodes, then let $A_0$ be one of those partitions for which $\lambda(A_0) = \lambda(f(A_0))$ if it exists; otherwise let $A_0$ be any of the largest alignment partitions. Initialize $M$ to include all nodes in $A_0$.}
\item{Repeatedly identify the largest alignment partition $A_i$ such that $\lambda(A_i)$ is not the label of any node in $M$ and $\lambda(f(A_i))$ is not the label of any node unified with a node in $M$, choosing $A_i$ to be one where $\lambda(A_i) = \lambda(f(A_i))$ if it exists; replace $M$ by $M \cup A_i$.}
\item{Stop when no more alignment partitions can be included in $M$.}
\end{compactenum}
If hash tables are used to record which node labels have been included in $M$ and in $f(M)$, checking for duplicate labels can be performed in $O(1)$ time. Partitions can be considered one by one in decreasing order of size, which requires $O(|T_q|\log(|T_q|))$ time to initialize and then $O(|T_q|)$ to enumerate since the number of partitions cannot exceed the number of nodes in $T_q$.

Finally, to compare a query SLT $T_q$ against a candidate SLT $T_c$, we choose a pair of aligned subtrees that maximizes the score for the candidate with respect to the query.

\begin{definition}{Maximum Subtree Similarity}
Given SLTs $T_q$ and $T_c$, consider pairs of aligned subtrees $T_{i_1}$ and $T_{i_2}$ as follows.
\begin{compactitem}
\item{The root of $T_{i_1}$ can be unified with the root of $T_{i_2}$.}
\item{$T_{i_1}$ is maximally similar to $T_q$.}
\end{compactitem}
The {\em Maximum Subtree Similarity} score MSS$(T_q,T_c)$ of $T_c$ with respect to $T_q$ is $\displaystyle \max_i s(T_q,T_c;T_{i_1},T_{i_2},M_i)$ over all such pairs, where $M_i$ is a greedily chosen matched set of nodes.
\end{definition}

\begin{theorem}
Computing Maximum Subtree Similarity for a candidate formula requires time $O(|T_c||T_q|^2\log(|T_q|))$.
\end{theorem}
\begin{proof}
The number of pairs of aligned subtrees is at most $|T_q|*|T_c|$. For each pair, checking whether the roots can be unified requires $O(1)$ time, checking maximum similarity requires $O(|T_q|)$ time, and computing the score requires constant time plus time $O(|T_q|log(|T_q|))$ to choose $M$.
\end{proof}

We show experimentally that this similarity metric performs very well.

\begin{table*}

\caption{NTCIR-11 Wikipedia Formula Retrieval Benchmark Results (100 Queries). For each system the top row shows \% recall (over all hits, top-$\infty$), and the bottom row shows the Mean Reciprocal Rank ~~(mrr, in \%). For Tangent-3, mrr for any formula identical to the target is also shown. }

\scriptsize
~\\
\centering
\begin{tabular}{ l | r | r r r r | r | r r r r }
  & \multicolumn{5}{c}{\sc Document-Centric} & \multicolumn{5}{c}{\sc Formula-Centric} \\
Participant & \sc Total & Easy & Frequent & Variable & Hard &  \sc Total & Easy & Frequent & Variable & Hard\\

\hline
\bf TUW Vienna & \bf 97 & 100 & 100 & 93 & 88 & \bf 93 & 100 & 96 & 89 & 63\\
~~(mrr) & \bf 82 & 97 & 50 & 96 & 54 & \bf 88 & 96 & 72 & 94 & 71\\

\hline

\bf NII Japan & \bf 97 & 98 & 100 & 93 & 100 & \bf 94 & 98 & 96 & 89 & 88\\
~~(mrr) & \bf 76 & 99 & 49 & 82 & 67 & \bf 77 & 89 & 92 & 78 & 48\\

\hline

{\bf Tangent-2 (RIT)} & \bf 88 & 98 & 79 & 89 & 63 & \bf 78 & 95 & 50 & 81 & 63\\
~~(mrr) & \bf 80 & 96 & 31 & 92 & 83 & \bf 86 & 94 & 47 & 96 & 83\\

\hline\hline

\multicolumn{8}{l}{\bf Tangent-3: Using exact formula location on target document} \\

{\bf w=1, EOL} & \bf 100 & 100 & 100 & 100 & 100 & \bf 89 & 95 & 67 & 100 & 88\\
~~(mrr) & \bf 83 & 100 & 55 & 95 & 41 & \bf 85 & 100 & 58 & 93 & 32\\

\hline
{\bf w=1 , No-EOL} & \bf 98 & 98 & 96 & 100 & 100 & \bf 87 & 93 & 63 & 100 & 88\\
~~(mrr) & 82 & 100 & 56 & 94 & 31 & \bf 84 & 100 & 59 & 92 & 32\\

\hline\hline

\multicolumn{8}{l}{\bf Tangent-3: Matching equivalent formulae on target document} \\

{\bf w=1, EOL} & &  &  & &  & \bf 100 & 100 & 100 & 100 & 100\\
~~(mrr) &  &  &  &  &  & \bf 82 & 100 & 55 & 93 & 28\\

\hline

{\bf w=1 , No-EOL} & & & & &  & \bf 98 & 98 & 96 & 100 & 100\\
~~(mrr) & & & & & & \bf 82 & 100 & 56 & 92 & 28\\

\hline
\end{tabular}\\

\label{tab:wikiresults}
\end{table*}

\section{Evaluation} \label{sec:experiments}

In this Section we present experiments designed to observe the effect of system parameters on index size, retrieval time, and search results. We do this using a combination of benchmarks, and a human experiment to evaluate the similarity of the Top-10 formulae returned by our system to query expressions.

{\bf Computational Resources.} We use a Ubuntu Linux 12.04.5 server with 24 Intel Xeon processors (2.93GHz) and 96GB of RAM. While some indexing operations were parallelized (as noted below), {\it all retrieval times are reported for a single process.}

\subsection{NTCIR-11 Formula Retrieval Benchmark} 

The NTCIR-11 Wikipedia benchmark \cite{wikipedia} is 2.5 GB, with 30,000 articles containing roughly 387,947 unique \LaTeX~expressions. The benchmark includes 100 queries for measuring specific-item retrieval performance, where each query is associated with a single target formula in a specific document. One or more wildcard symbols are present in 35 of the queries.
{\it Easy} queries (41) and {\it Frequent} queries (24) have no wildcard symbols, and are distinguished by whether one or multiple formulae in the corpus match the target expression. {\it Variable} queries (27) and {\it Hard} queries (8) contain wildcards, and are again distinguished by whether one or multiple formulae match the target. Search results are returned as a ranked list of $(\mathit{documentId}, \mathit{formulaId})$ pairs.

Systems are evaluated using two metrics. First, by the percentage of targets located (at any rank), and second by the Mean Reciprocal Rank (mrr) of successfully retrieved targets. The `document-centric' evaluation filters results so that document identifiers are listed in their order of appearance in the results. The `formula-centric' results are computed using the complete ranked list of matches. Previously, `formula-centric' results were computed using specific formula identifiers, so for example, given a query and target formula, if the target formula is found at a different location within the target document, this is considered a miss.

At the top of Table \ref{tab:wikiresults}, NTCIR-11 results from the two best systems and from Tangent version 2 are shown as fractions rounded to the nearest percentage \cite{wikipedia}.  The results for Tangent-3 are shown below this in Table \ref{tab:wikiresults}. For the formula-centric evaluation, we present results when defining hits using specific formula identifiers, and when accepting any identical formula on the target document.
The large difference with these two definitions of hits indicates that our system should index all occurrences of formulae in documents, rather than just the first unique occurrence of each formula in a document as done currently.

For Tangent-3, all combinations of window sizes $w=\{1,2,3,4,\infty\}$ (where $\infty$ is all tuples) and including or excluding end-of-line symbols (EOL/No-EOL) were used. Surprisingly, different window sizes had very little effect on performance, although adding EOL tuples increased the number of formulae retrieved by two (e.g. the query  `s' could be located). We show results for windows size 1 ($w=1$) and with and without end-of-line symbols in Table \ref{tab:wikiresults}.

Tangent-3 obtains the best document recall and mrr results to date, with perfect recall when using EOL tuples. The slightly lower {\it Variable} and {\it Hard} query mrr values are an artifact of additional challenging formulae being located at lower ranks. For formula evaluation, the formula recall is improved over Tangent-2 when using exact formula id matches for hits, but when we treat equivalent formula on a document as a match, we again obtain perfect recall using EOL tuples. As before, the mrr is slightly lower than in some other results because additional formulae have been found which are located at relatively low ranks.

From a user's perspective, after formulae are grouped according to their SLT matches (see Figure \ref{fig:q14}), 89 of the target formulae appear in the first (Top-1) group, and 95 within the Top-3 groups. If EOL symbols are added, then 97 of the queries are found in the Top-3 formula groups.

\subsection{Indexing and Retrieval}

We used both NTCIR-11 collections to test the index sizes and retrieval times for our system. In addition to the Wikipedia collection described above, we use the much larger NTCIR-11 arXiv collection to test the scalability of Tangent. The arXiv collection is 174 GB uncompressed, with 8,301,578 documents (fragments from arXiv articles) and roughly 60 million formulae including isolated symbols.

{\bf Indexing Time.} The arXiv data took 43 hours to pre-process the documents (using 10 processes), and at most an additional 3.5 hours to generate the index using a single process (all tuples, with end-of-line tuples). Wikipedia was much faster, requiring 260 seconds for preprocessing, and at most 95 seconds for index creation. As our document pre-processor and re-ranker are implemented in Python, we believe that a faster implementation (e.g., in C++) could reduce run times by a factor of 4-10 in both cases.

{\bf Index Sizes.} As seen in Table \ref{tab:index}, when all tuples are stored with EOL tuples, the Wikipedia index is 499 MB on disk; in contrast, for the arXiv this maximum index size is 29 GB. When these index files are loaded into memory, they consume 2 - 2.5 times their space on disk. Index size increases roughly linearly from window sizes 1-4, with end-of-line tuples increasing storage by a constant amount. For smaller window sizes storage is much smaller; for $w=1$ without end-of-line tuples the index file is 63 MB for Wikipedia, and 5.2 GB for arXiv; these are much smaller than Tangent-2 (1.3 GB for Wikipedia, and roughly 36 GB for the arXiv dataset~\cite{Pattaniyil2014}).

{\bf Retrieval and Reranking Times.} We ran all 100 of the NTCIR-11 Wikipedia queries over the arXiv collection to
test retrieval speed. Retrieval is now much faster than Tangent-2 (see Section \ref{sec:problem}). We see in Figure~\ref{fig:goldstandard}b that median retrieval times are less than 1 second in all conditions, but much faster without end-of-line symbols. For the smaller Wikipedia collection retrieval is much faster, with average retrieval time without EOL at $w=1$ being 7ms ($\sigma=23ms$) and 9ms at $w=2$ ($\sigma=31ms$).

Re-ranking times are consistent across core parameters because the number of candidates reranked is fixed at $k=100$. For Wikipedia, re-rank times are (median, $\mu$, $\sigma$) = $(72, 775, 3562)$ ms. The mean is skewed significantly by a small number of outliers. In one extreme case, re-ranking took 46 seconds, while retrieval from the core ($w=\infty$ with EOL) was just 1.7 seconds.  This particular expression is very large, with 16  wildcard symbols (\emph{Query 52}), producing large hits with many possible unifications. We do not expect to see many queries of this type in common use.

\subsection{Evaluation of MSS Reranking}

\begin{table}
\caption{Index Sizes for NTCIR-11 Collections}
\scriptsize
~\\
\centering
\begin{tabular}{r | r r | r r l}
\hline
~&
\multicolumn{4}{c}{\bf Index Sizes ($MB$)} \\
& \multicolumn{2}{c}{\sc Wikipedia} & \multicolumn{2}{c}{\sc arXiv}\\
w  &   No-EOL & EOL &   No-EOL & EOL\\
\hline

1 & 63.1 & 72.6 & 5,238 & 6,036 \\
2 & 94.4 & 103.9 & 7,419 & 8,216 \\
3 & 126.9 & 136.4 & 9,491 & 10,288 \\
4 & 159.7 & 169.1 & 11,397 & 12,194 \\
$\infty$ & 489.2 & 498.7 & 28,099 & 28,897 \\
\hline

\end{tabular}
\label{tab:index}
\end{table}

\begin{figure}[!tb]
\small
\centering
\begin{tabular}{c c}

\scalebox{0.25}{\includegraphics[trim=32 15 30 20,clip]{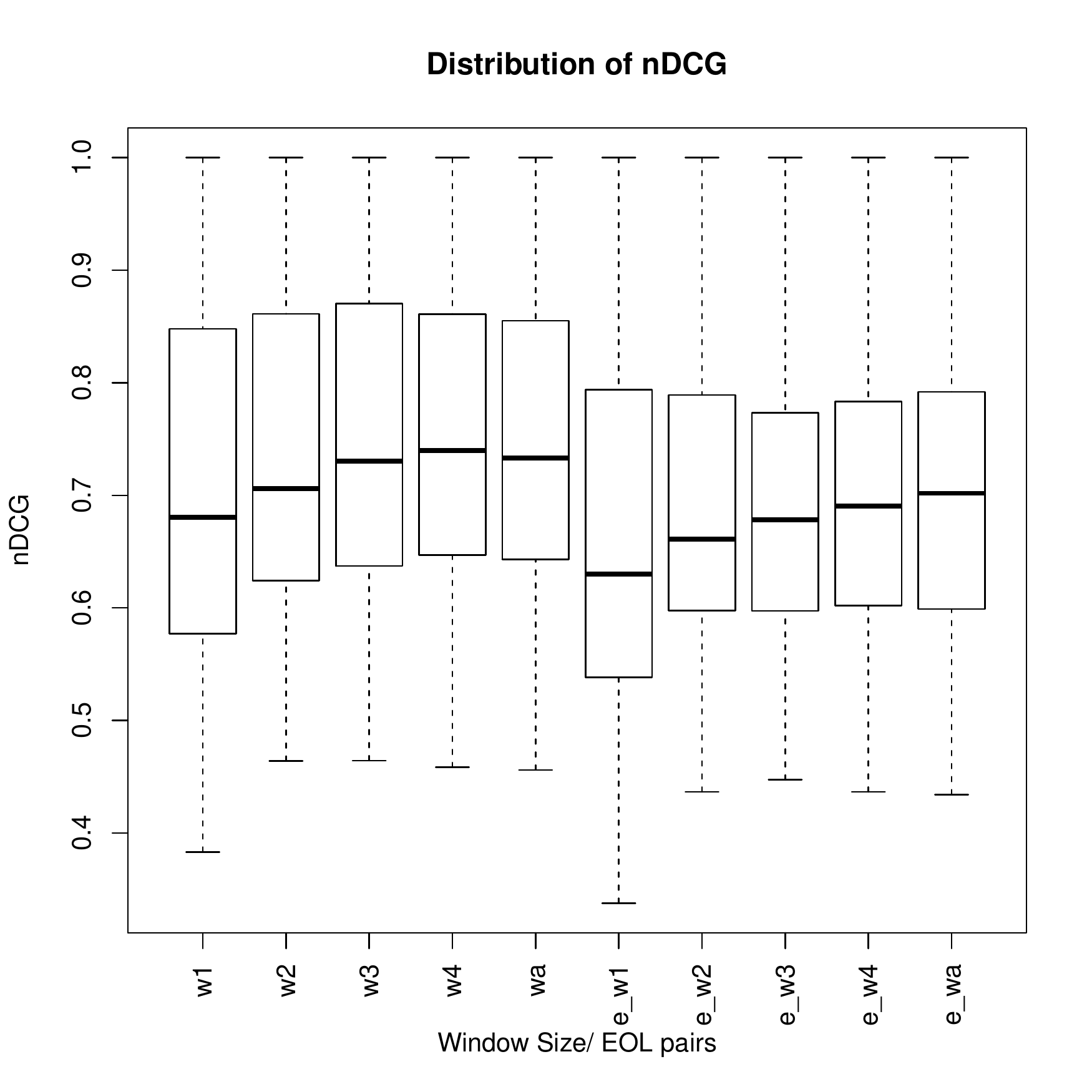}} &
\scalebox{0.25}{\includegraphics[trim=32 15 30 20,clip]{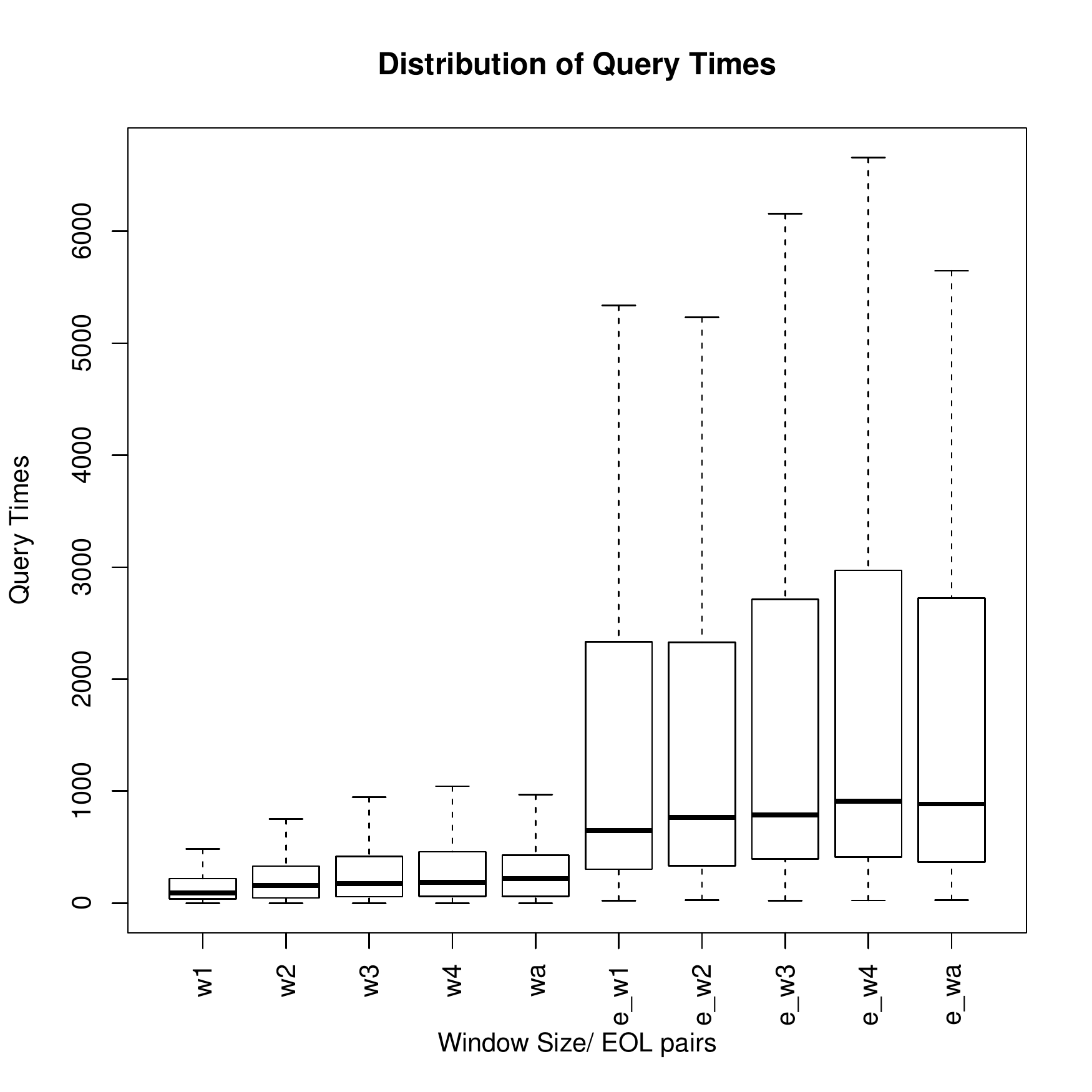}} \\
a) MSS Distributions (nDCG)  &
b) Retrieval times (ms)\\
\end{tabular}

\caption{Distribution of Top-100 nDCG (MSS) Scores and Wikipedia Query Retrieval Times for the NTCIR-11 arXiv Collection.}

\label{fig:goldstandard}
\end{figure}

We now consider how well MSS-based rankings correspond to human perceptions of formula similarity, through evaluating Top-10 results.  

To first select which combinations of parameters to use for our human evaluation, we examined the MSS scores of formulae returned by the core before re-ranking.  In Figure \ref{fig:goldstandard}a, from the Top-100 hits returned by the core for the NTCIR-11 Wikipedia task, we compute normalized Discounted Cumulative Gain (nDCG@100) distributions for the Maximum Subtree Similarity Scores in each of the Top-100 hits compared to an MSS `gold standard.' In the gold standard all formulae in the Wikipedia collection have been scored for each of the 100 Wikipedia queries, and the top-100 formulae for each query are used for normalization. We again consider a number of different window size and EOL parameters ($w=\{1,2,3,4,\infty\}$, with and without EOL tuples). The first five columns show increasing $w$ values without EOL tuples, followed by the same range of $w$ values with EOL tuples.

 Adding EOL tuples shifts values around the median down. We took this as evidence that including EOL tuples was not helping return more similar formulae as measured by nDCG over the MSS scores. Further, as moving from $w=1$ to $w=2$ reduces the variance most dramatically, to keep the number of hits for individual participants to evaluate reasonable, we chose to consider only $w=1, 2,$ and $\infty$.

\subsubsection{Experimental Design}

{\bf Data.} A set of 10 queries were selected using random sampling from the Wikipedia query set. Five of the queries contained wildcards, and the other five did not. 
Some queries were manually rejected and then randomly replaced to insure that a diverse set of expression sizes and structures were collected. Using the Wikipedia collection, for the three versions of the core compared ($w=\{1,2,\infty\}$, no EOL tuples), we applied reranking to the Top-100 hits, and then collected the Top-10 hits returned by each query for rating.

{\bf Evaluation Protocol.} Participants completed the study alone in a private, quiet room with a desktop computer running the evaluation interface in a web browser. The web pages first provided an overview, followed by a demographic questionnaire, instructions on evaluating hits, and then familiarization trials (10 hits; 5 for each of two queries).  After familiarization participants evaluated hits for the 10 queries, 
and finally completed a brief exit questionnaire. Participants were paid \$10 at the end of their session.

Participants rated the similarity of queries to results using a five-point Likert scale (Very Dissimilar, Dissimilar, Neutral, Similar, Very Similar). It has been shown that presenting search results in an ordered list  affects the likelihood of hits being identified as relevant~\cite{guan2007eye}. Instead we presented queries along with each hit in isolation. To avoid other presentation order effects, the order of query presentation was randomized, followed by the order in which hits for each query were presented.

\subsubsection{Results}

{\bf Demographics and Exit Questionnaire.}
21 participants (5 female, 16 male) were recruited from the Computing and Science colleges at our institution. Their age distribution was: 18-24 (8), 25-34 (9), 35-44 (1), 45-54 (1), 55-64 (1) and 65-74 (1). Their highest levels of education completed were: Bachelor's degree (9), Master's degree (9), PhD (2), and Professional Degree (1).  Their reported areas of specialization were: Computer Science (13), Electrical Engineering (2), Psychology (1), Sociology (1), Mechanical Engineering (1), Computer Engineering (1), Math (1) and Professional Studies (1).

In the post-questionnaire, participants rated the evaluation task as Very Difficult (3), Somewhat Difficult (10), Neutral (6), Somewhat Easy (2) or Very Easy (0). They reported different approaches to assessing similarity. Many considered whether operations and operands were of the same type or if two expressions would evaluate to the same result. Others reported considering similarity primarily based on similar symbols, and shared structure between expressions.

\begin{table}
\caption{Mean and Standard Deviation Likert Ratings for Top-10 NTCIR-11 Wikipedia Hits  (21 Participants, 10 queries).}
\centering
\tiny
\begin{tabular} {c | l l l l l}
\hline
& \multicolumn{5}{c}{\sc Rank/Position in Top-10 Hits}\\
& {\bf 1} & {\bf 2} & {\bf 3} & {\bf 4} & {\bf 5} \\
\hline
$w=1$ & 4.54 (0.78) & 3.79 (1.16) & 3.48 (1.31) & 3.20 (1.30) & 2.83 (1.24) \\
$w=2$ & 4.54 (0.78) & 3.71 (1.22) & 3.48 (1.30) & 3.16 (1.28) & 2.90 (1.25) \\
$w=a$ & 4.54 (0.78) & 3.78 (1.18) & 3.59 (1.22) & 3.27 (1.15) & 2.98 (1.24) \\
\hline
& {\bf 6} & {\bf 7} & {\bf 8} & {\bf 9} & {\bf 10} \\
\hline
$w=1$ & 2.94 (1.22) & 2.65 (1.19) & 2.78 (1.20) & 2.78 (1.20) & 2.85 (1.24) \\
$w=2$ & 2.93 (1.19) & 2.85 (1.25) & 2.57 (1.18) & 2.74 (1.22) & 2.80 (1.13) \\
 $w=a$ & 2.92 (1.23) & 2.80 (1.17) & 2.98 (1.23) & 2.92 (1.21) & 2.87 (1.17)\\

\hline
\end{tabular}

\label{likert}
\end{table}

{\bf Similarity Ratings.} As seen in Table~\ref{likert}, the Likert-based similarity rating distributions are 
very similar, and identical in a number of places. In all three conditions, average ratings increase consistently from the 5th to 1st hits. The top 4 formula hits all have an average rating higher than '3,' suggesting that a number of participants felt these formula had some similarity with the query expression. After this the ratings are less than `3' and sometimes shift. Perhaps because matches were not highlighted, in a number of cases exact matches were rated as `4' rather than `5.' As was found for the NTCIR-11 Wikipedia benchmark, it appears that a window size of 1 is able to obtain
strong results. This is appealing, because this requires the smallest index size and has the fastest retrieval times.

\section{Conclusion}

We have presented a new technique for ranking appearance-based formula retrieval results,
using the candidate formula subtree with the
harmonic mean for matched symbols and edges after greedy unification of symbols by type.
This Maximum Subtree Similarity (MSS) metric prefers large connected matches of the query
within the formula. In an experiment we found that for the Top-10 hits, the human ratings of similarity were consistent with the ranking produced by our metric. We have also described an efficient two-stage
implementation of our retrieval model that produces state-of-the-art results for the NTCIR-11 Wikipedia
formula retrieval task, using a much smaller index.

In the future we plan to explore using end-of-line symbols, but only for small expressions.
This will not require much additional space in the index, while greatly reducing the cost of
wildcard end-of-line tuples. 
We also plan to support multiple copies of a formula in a document, devise new methods for
ranking documents based on multiple matches and/or query expressions, and integrate our
formula retrieval system with keyword search.

\section*{Acknowledgements}

This material is based upon work supported by the National Science Foundation (USA) under Grant No. IIS-1016815 and HCC-1218801. Financial support from the Natural Sciences and Engineering Research Council of Canada under Grant No. 9292/2010, Mitacs, and the University of Waterloo is gratefully acknowledged.

\bibliographystyle{abbrv}


\end{document}